\documentclass[prd,aps,amsfonts,superscriptaddress,nofootinbib,longbibliography,notitlepage,twocolumn]{revtex4-1}

\usepackage{graphicx}
\usepackage{xcolor}
\usepackage{rotating}
\usepackage{amsmath,amssymb,graphics,amsthm,isomath}

\usepackage[colorlinks=true, urlcolor=violet, linkcolor=blue, citecolor=red, hyperindex=true, linktocpage=true]{hyperref}
\usepackage[capitalise,compress]{cleveref}

\newtheorem{thm}{Theorem}
\newtheorem{cor}[thm]{Corollary}
\newtheorem{lem}[thm]{Lemma}

\makeatletter
\renewcommand{\p@subsection}{}
\renewcommand{\p@subsubsection}{}
\makeatother

\usepackage{xcolor}
\usepackage{mathtools}

\usepackage{dsfont}

\begin{document}

\title{Bound on quantum scrambling with all-to-all interactions}

\author{Chao Yin}
\email{chao.yin@colorado.edu}
\affiliation{Department of Physics, Peking University, Beijing 100871, China}
\affiliation{Department of Physics, University of Colorado, Boulder CO 80309, USA}

\author{Andrew Lucas}
\email{andrew.j.lucas@colorado.edu}
\affiliation{Department of Physics, University of Colorado, Boulder CO 80309, USA}
\affiliation{Center for Theory of Quantum Matter, University of Colorado, Boulder CO 80309, USA}

\begin{abstract}
We prove bounds on operator growth and infinite temperature out-of-time-ordered correlators in many-body systems with $N$ spin-$\frac{1}{2}$ degrees of freedom which interact via two-body all-to-all interactions.  
Our results parametrically improve previous bounds, and sharply constrain when and how quantum simulators, including trapped ion crystals and cavity quantum electrodynamics,  can study quantum gravity.

\end{abstract}

\date{\today}

\maketitle

\section{Introduction}
In the past few years, the holographic correspondence between quantum many-body systems and quantum gravity in one more spacetime dimension \cite{Maldacena:1997re} has attracted intense interest.  In particular, the realization that microscopic models including the Sachdev-Ye-Kitaev model \cite{Sachdev:2015efa,Maldacena:2016hyu,Kitaev:2017awl} might realize (aspects of) quantum gravity has set off a hunt for microscopic models that mimic quantum gravity, and might also be studied experimentally \cite{Chew_2017,Chen_2018,Marino_2019,Lewis_Swan_2019,Alavirad_2019,Bentsen_2019prl,Bentsen_2019prx}.

A key property of quantum black holes (and thus a theory of quantum gravity) is that they are \emph{fast scramblers} \cite{Sekino:2008he}.  For our purposes, fast scrambling means that out-of-time-ordered correlators (OTOCs) exhibit exponential growth \cite{Shenker:2013pqa}. In a theory of $N$ spin-$\frac{1}{2}$ degrees of freedom, we expect that for any $1\le i,j\le N$, \begin{equation}
    \left\langle [X_i(t),X_j]^2 \right\rangle \sim -\frac{1}{N}\mathrm{e}^{\lambda t}. \label{eq:toyotoc}
\end{equation}
Here $X_i$, $Y_i$ and $Z_i$ denote the three Pauli matrices acting on spin $i$.  The key feature of (\ref{eq:toyotoc}) is that $\lambda$ is independent of $N$, and so the OTOC (as we defined it) becomes of order 1 in a \emph{scrambling time} \begin{equation}
    t_{\mathrm{s}} \sim \log N. \label{eq:tslogN}
\end{equation}
(\ref{eq:tslogN}) is believed to hold in all theories of quantum gravity, for (almost?) every pair of $i$ and $j$.  Note that (\ref{eq:tslogN}) serves as our (informal) definition of scrambling time.

The canonical Lieb-Robinson theorem \cite{Lieb1972}, which says that quantum information spreads ballistically in a $d$-dimensional lattice, forbids fast scrambling in conventional lattice models.  However, generalizations of the Lieb-Robinson bounds to spin systems defined on more abstract interaction graphs, including those with all-to-all interactions (each spin couples to each other spin), do suggest fast scrambling is permitted \cite{Lashkari_2013,Bentsen_2019,guo2019signaling}.  

Happily, it is experimentally possible to realize the spatially non-local interactions required of a fast scrambler.  As a simple example, we can realize the Hamiltonian \begin{equation}
    H = \frac{1}{N^\alpha} \sum_{i,j=1}^N J_{ij}(t)Z_i Z_j + \sum_{i=1}^N \sum_{A=1}^3 h^A_i(t) X^A_i \label{eq:H}
\end{equation}
in a trapped ion crystal \cite{Britton_2012} (with expected exponent $\alpha=1$).   Hamiltonians with similar simple all-to-all interactions can be achieved in cavity quantum electrodynamics \cite{Leroux_2010,thompson2020} (with $\alpha=0$).  Remarkably, in these platforms it is possible to measure certain kinds of OTOCs \cite{Garttner_2017, Li_2017,wei2019OTOC}. Here $X^A_i = \lbrace X_i,Y_i,Z_i\rbrace$ is shorthand for the three Pauli matrices, and $\alpha$ is a free parameter governing the strength of the all-to-all interactions; we take $J_{ij}$ to scale independently of $N$. In the simplest experiments, all $J_{ij}=J$ are the same. If such a system can model a fast scrambler, it would allow for near-term experimental tests of aspects of quantum gravity.

Our goal is to rigorously address the extent to which (\ref{eq:H}), along with many generalizations, could realize a fast scrambler in an experiment.  We will show that at infinite temperature, in this family of models (\ref{eq:H}), \begin{equation}
    t_{\mathrm{s}} \gtrsim N^{\alpha - \frac{1}{2}}. \label{eq:main}
\end{equation}
(We postpone the precise statement and its proof.) Hence it is impossible to have both fast scrambling (which requires $\alpha\le \frac{1}{2}$) and extensivity of the energy spectrum ($\alpha\ge 1$), at least in the  model with $J_{ij}=J$.

\section{Implications}
Our main result (\ref{eq:main}) is complementary to recent works \cite{li2020fast,belyansky2020minimal} which have proposed studying fast scrambling in models of a similar form to (\ref{eq:H}).  Our bound (\ref{eq:main}) is not incompatible with their key results, so long as $\alpha\le \frac{1}{2}$ is taken. Whether such a small value of $\alpha$ has further interesting consequences or constraints on the many-body dynamics is an interesting open question. It is unclear whether this constraint is irrelevant for the faithful simulation of quantum gravity in an experiment. We do note, however, that our bound (\ref{eq:main}) may not be tight in some models, in which case (\ref{eq:tslogN}) might be achieved even when $\alpha<\frac{1}{2}$ \cite{li2020fast}.  However, we will conclude this paper by showing that it is possible to violate (\ref{eq:tslogN}) in at least one model with $\alpha<\frac{1}{2}$.  Our bound (\ref{eq:main}) cannot be parametrically tightened (at least when $\alpha\ge \frac{1}{2}$).

Indeed, (\ref{eq:main}) has clear implications for how a quantum simulator, such as a trapped ion crystal, could be used for the experimental study of information scrambling and quantum gravity.  Certainly we must take $\alpha\le \frac{1}{2}$ to realize fast scrambling.  The only way for such a model to be thermodynamically extensive is for $J_{ij}$ to be a matrix with order 1 entries and maximal eigenvalue $N^{1/2}$.  Heuristically this means that $J_{ij}$ is a random matrix \cite{mehta}.  Unfortunately, such a regime is not yet realized in a coherent quantum simulation with hundreds of qubits.  For example, focusing on trapped ion platforms, such a regime would require detuning the driving laser very far from the vibrational modes of the ion crystal \cite{Britton_2012}, leading to very weak collective interactions.

The more practical alternative for experiments is to use the single-site fields in (\ref{eq:H}) to dephase the many-body wave function, leading to genuine quantum dynamics and scrambling.  After all, since all Pauli $Z$s commute in (\ref{eq:H}), having single-site $X$ and $Y$ fields is mandatory to realizing chaos.  The results of \cite{belyansky2020minimal} suggest this approach may be feasible.  However, the Hamiltonian must then be strongly time-dependent, meaning that no finite temperature physics may be realized. As the emergence of a semiclassical bulk geometry out of quantum dynamics critically relies on a low temperature compared to microscopic energy scales, many questions about quantum gravity may be inaccessible.

Our bound (\ref{eq:main}) is parametrically stronger than existing Lieb-Robinson bounds \cite{Lashkari_2013, Bentsen_2019,guo2019signaling}.  We derived it using a more general operator growth formalism developed in \cite{chen2019operator,chen2019finite,lucas2019nonperturbative,tran2020hierarchy}, which relies on the simple relation between OTOCs and operator size at infinite temperature \cite{nahum_operator_2018,von_keyserlingk_operator_2018,Roberts:2018mnp}. 
The Lieb-Robinson bounds of \cite{Lashkari_2013, Bentsen_2019,guo2019signaling} might be saturated by studying OTOCs prepared in finely tuned initial states (the infinite temperature ensemble measures the value of the correlator in a typical state).  Indeed, \cite{tran2020hierarchy} recently discovered that there are two separate notions of locality that arise in systems with power law interactions;  it would not be surprising if a similar phenomenon arose in models with all-to-all interactions.

It is worth keeping in mind that ``fast scrambling" is not necessarily the ``fastest scrambling" in nature \cite{Bentsen_2019,lucas2019star}.  It is possible to find models with $N$-independent scrambling times, which are certainly not holographically dual to quantum gravity.  In particular, it is possible to find quantum mechanical models, defined on ``star-like graphs" where $N$ spins are coupled to a single spin, where the OTOC in (\ref{eq:toyotoc}) becomes order 1 in an $N$-independent time, violating (\ref{eq:tslogN}).  We do not expect that such models are described by a holographic dual, as all known holographic models exhibit the exponential growth of (\ref{eq:toyotoc}), both at finite temperature and infinite temperature. In the future, we hope to find further-refined probes of holography and quantum gravity to better discriminate between holographic and non-holographic models with non-local interactions.

Finally, we emphasize that another important feature of models of quantum gravity is that at finite temperature $T$, the Lyapunov exponent $\lambda$ in (\ref{eq:toyotoc}) is bounded: $\lambda \le 2\pi T$ \cite{Maldacena:2015waa}.  Extending our work to finite temperature is a critical (and quite challenging) technical problem.  Nevertheless, as emphasized above, we do not expect a model which is not a fast scrambler at $T=\infty$ (in the sense of (\ref{eq:toyotoc})) will become a fast scrambler at finite $T$. 

\section{Formal discussion} 
The remainder of this paper consists of the proof of (\ref{eq:main}).  First, we make precise our assumptions and state a theorem; we conclude with its proof. We study quantum many-body systems consisting of $N$ spin-$\frac{1}{2}$ degrees of freedom.  The spins are labeled by vertices $v$ in the set $V$.  The Hilbert space $\mathcal{H}$ is (isomorphic to) $(\mathbb{C}^2)^{\otimes N}$.  As above, $X_i$, $Y_i$ and $Z_i$ denote the Pauli matrices (normalized as $X_i^2=1$) on spin $i$ ($i\in V$).  

In a nutshell, we will prove our bound on OTOC growth by interpreting $\langle [X_i(t),X_j]^2\rangle$ as the length of a vector $[X_i(t),X_j]$.  In fact, this is natural, since the linearity of quantum mechanics means that $[X_j,\cdot]$ is a linear transformation on a vector space of operators acting on Hilbert space, and that $X_i(t)$ is obtained by applying an appropriate linear transformation to the vector $X_i$.  The inner product which calculates the length of these vectors is the Frobenius norm.  This approach will allow us to obtain much stronger bounds than Lieb-Robinson bounds, which control the operator norm (and are not useful for sharply bounding the OTOC). 

Let $\mathcal{B}$ denote the set of Hermitian operators acting on $\mathcal{H}$.  It is spanned by products of Pauli matrices on every qubit, along with the identity: \begin{equation}
    \mathcal{B} = \bigotimes_{i=1}^N \mathcal{B}_i=\bigotimes_{i=1}^N \lbrace 1, X,Y,Z\rbrace_i. \label{eq:basisB}
\end{equation}
We denote elements of $\mathcal{B}$ by $|\mathcal{O})$ -- these are just like Dirac kets, but with a parentheses to emphasize the vector space is $\mathcal{B}$, not $\mathcal{H}$.  The appropriate inner product on $\mathcal{B}$ for studying infinite temperature chaos is \begin{equation}
    (A|B) = 2^{-N}\mathrm{tr}(A^\dagger B).
\end{equation}
Note that $(A|A)$ is, in general ,much smaller than $\lVert A\rVert^2$, where the latter represents the conventional (squared) operator norm, which for a Hermitian matrix returns the maximal eigenvalue (squared).  The basis vectors of (\ref{eq:basisB}) are orthonormal.  Time translation on $\mathcal{B}$ is generated by the Liouvillian \begin{equation}
    \mathcal{L}(t) = \mathrm{i}[H(t),\cdot]. \label{eq:liouvillian}
\end{equation}
$\mathcal{L}(t)$ is an antisymmetric linear operator on $\mathcal{B}$, and \begin{equation}
    \frac{\mathrm{d}}{\mathrm{d}t} |\mathcal{O}(t)) = \mathcal{L}(t)|\mathcal{O}(t)).
\end{equation}

Define the projection operation \begin{equation}
    \mathbb{P}_i |\mathcal{O}) = |\mathcal{O}) - \frac{1}{2}\left| 1_i \otimes \mathrm{tr}_i \mathcal{O}\right)
\end{equation}
where $\mathrm{tr}_i$ denotes partial trace over qubit $i$.  This operation removes all products of Pauli matrices which include the identity on site $i$.  Clearly, the infinite temperature OTOC obeys \begin{align}
    \left|2^{-N}\mathrm{tr}\left([X_i(t),X_j]^2\right)\right| &=  \left|2^{-N}\mathrm{tr}\left([\mathbb{P}_jX_i(t),X_j]^2\right)\right| \notag \\
    &\le 4 (X_i(t)|\mathbb{P}_j|X_i(t)). \label{eq:pjotoc}
\end{align}
This conclusion generalizes to allow for $X_i$ and $X_j$ to be any linear superposition of Paulis.  For any subset $S\subseteq V$ similarly define $\mathbb{P}_S$ to be the projection onto all operators which have at least one non-identity Pauli on at least one vertex $i\in S$.

Let $0<a<1$ be an $N$-independent constant.  We \emph{define} the scrambling time as the smallest possible time $t_{\mathrm{s}}>0$ at which the projection in (\ref{eq:pjotoc}) is large: \begin{equation}
   t_{\mathrm{s}} = \inf_{t\in\mathbb{R}^+} \left\lbrace \sup_{\mathcal{O}_i\in\mathcal{B}_i} \sum_{j=1}^N \frac{(\mathcal{O}_i(t)|\mathbb{P}_j|\mathcal{O}_i(t))}{(\mathcal{O}_i|\mathcal{O}_i)} > aN\right\rbrace. \label{eq:tsdef}
\end{equation} Our key conclusions do not depend on $a$.  The quantity in the sum above is called the \emph{average operator size} \cite{nahum_operator_2018,von_keyserlingk_operator_2018,Roberts:2018mnp}.

Formally, we say that a graph $\Lambda=(V,E)$ exists in $d$ spatial dimensions ($\Lambda$ is $d$-dimensional) when the following properties hold.  Let $N=|V|$ be the number of vertices.  Pick any vertex $v\in V$.  Define $S_D$ to be the set of all vertices that are a distance $\le D$ away from $v$: namely, for any $x\in S_D$ a path of at most $D$ edges exists from $v$ to $x$.  We say $\Lambda$ is $d$-dimensional when there exist finite constants $0<c_1,c_2<\infty$ that are independent of $N$, such that for every $v$, 
\begin{subequations}\begin{align}
    |S_D| &\le c_1 D^d, \label{eq:SDbound}\\
    |S_D| - |S_{D-1}| &\le c_2 D^{d-1}.
\end{align}
\end{subequations}
We are now ready to state our main result: 
\begin{thm}
Let $\Lambda=(V,E)$ be a $d$-dimensional lattice graph with $|V|=N$ vertices, such that each vertex in the graph has at most $k$ vertices, with $k$ finite and independent of $N$.  Consider quantum dynamics on $\mathcal{H}=(\mathbb{C}^2)^N$ generated by  \begin{align}
    H(t) &= \sum_{i,j \in V} \frac{J_{ij}^{AB}(t)}{N^\alpha} X^A_iX^B_j + \sum_{\lbrace i,j\rbrace \in E} K^{AB}_{ij}(t) X^A_i X^B_j \notag \\
    &+ \sum_{i\in V} h^A_i(t) X^A_i \label{eq:thmH}
\end{align}
where $|J_{ij}^{AB}(t)|\le 1$ and $|K^{AB}_{ij}(t)| \le 1$. We employ Einstein summation convention on $A,B$ indices.  Then there exists $0<C<\infty$ such that if $\alpha < 1+\frac{1}{d}$, \begin{equation}
   t_{\mathrm{s}} > C N^{(2\alpha-1)/(d+2)} . \label{eq:thm}
\end{equation}
In other words, for any $\alpha>\frac{1}{2}$, the Hamiltonian (\ref{eq:thmH}) is not a fast scrambler.  For $\alpha \ge 1+\frac{1}{d}$, $t_{\mathrm{s}} > CN^{1/d}$ is not affected by the non-local interactions.
\end{thm}

\section{Proof of the Theorem}
For notational simplicity, we assume below that $H(t)$ does not depend on time.  However, the proof below immediately generalizes to the time-dependent case, which we ``leave as an exercise to the reader".

The strategy of the proof is as follows: we aim to bound how quickly an operator $|X_i(t))$ rotates into a space of ``large" operators, supported far from site $i$.  What we will see is that due to the Frobenius norm controlling operator length, the all-to-all interaction rotates a single Pauli $X_i$ into a sum of Pauli strings of length $\sqrt{N} \times N^{-\alpha}$; the factor of $\sqrt{N}$ comes from the fact that $X_iX_j$ is orthogonal to $X_iX_k$ when $j\ne k$, and so on: it is the \emph{squared lengths} which should be added together.  In contrast, a Lieb-Robinson bound would use the fact that the operator norm of $X_i(X_j+X_k)$ can be additive.  To obtain (\ref{eq:thm}), we need to further account for operator growth due to the local interactions; however, we can also see a hint for the critical value of  $\alpha=\frac{1}{2}$ from our simple argument.

We now provide the details of the proof.   Choose any vertex $v\in V$,  $D \in \mathbb{Z}^+$, and let $\bar S_D$ denote the complement of $S_D$.   Now let us define \begin{subequations}
\begin{align}
    H_{<D} &:= \sum_{\lbrace i,j\rbrace \subset S_D} K^{AB}_{ij} X^A_i X^B_j + \sum_{i\in S_D} h^A_i X^A_i, \\
    H_D &:= \sum_{i \in S_D, j\notin S_D} K^{AB}_{ij}X^A_i X^B_j, \\
    H_{>D} &:= \sum_{\lbrace i,j\rbrace \subset \bar S_D} K^{AB}_{ij} X^A_i X^B_j + \sum_{i\in \bar S_D} h^A_i X^A_i, \\
    H_{\mathrm{<NL}} &:= \left(\sum_{i\in S_D, j\notin S_D} + \sum_{\lbrace i,j\rbrace \subset S_D}\right) \frac{J_{ij}^{AB}}{N^\alpha} X^A_iX^B_j, \\
    H_{\mathrm{>NL}} &:= \sum_{\lbrace i,j\rbrace \subset \bar S_D} \frac{J_{ij}^{AB}}{N^\alpha} X^A_iX^B_j.
\end{align}
\end{subequations}
$H_{>D}$ and $H_{\mathrm{>NL}}$ are the terms in the Hamiltonian that do not act on vertices in $S_D$; $H_{<D}$ acts entirely within $S_D$; $H_D$ and $H_{<\mathrm{NL}}$ denote terms which connect $S_D$ and $\bar S_D$.  We define $\mathcal{L}_{<D}(t)$, etc., in the obvious way, using (\ref{eq:liouvillian}).  Note that \begin{equation}
    H = H_{<D} + H_D + H_{>D} + H_{\mathrm{<NL}}+H_{\mathrm{>NL}}.
\end{equation}  

Following \cite{chen2019operator,chen2019finite,lucas2019nonperturbative,tran2020hierarchy}, we invoke the Duhamel identity \begin{equation}
    \mathrm{e}^{\mathcal{L}t} = \mathrm{e}^{\mathcal{L}_{<D}t} + \int\limits_0^t\mathrm{d}s \; \mathrm{e}^{\mathcal{L}(t-s)} (\mathcal{L}-\mathcal{L}_{<D}). \mathrm{e}^{\mathcal{L}_{<D}s} \label{eq:duhamel}
\end{equation}
Let $|\mathcal{O}_v)$ denote a (linear combination of) Paulis on vertex $v$ with $(\mathcal{O}_v|\mathcal{O}_v) = 1$.  Since $|\mathcal{O}_v(t)) = \mathrm{e}^{\mathcal{L}t}|\mathcal{O}_v)$, we can apply (\ref{eq:duhamel}).  Now, how much of $|\mathcal{O}_v(t))$ has support in $\bar S_D$?  Observe that 
\begin{equation}
\mathbb{P}_{\bar S_D} |\mathcal{O}_v(t)) =  \mathbb{P}_{\bar S_D} \int\limits_0^t\mathrm{d}s \; \mathrm{e}^{\mathcal{L}(t-s)} (\mathcal{L}_D+\mathcal{L}_{\mathrm{<NL}}) \mathrm{e}^{\mathcal{L}_{<D}s} | \mathcal{O}_v).
\end{equation} 
Since $\mathrm{e}^{\mathcal{L}t}$ and $\mathrm{e}^{\mathcal{L}_{<D}t}$ are unitary transformations, they do not change the length of $|\mathcal{O}_v)$ as measured by our inner product.  Thus, we use the triangle inequality to obtain \begin{align}
    \left\lVert\mathbb{P}_{\bar S_D} |\mathcal{O}_v(t))\right\rVert_2 &\le \int\limits_0^t \mathrm{d}s  \left\lVert \mathcal{L}_D \mathrm{e}^{\mathcal{L}_{<D}s} | \mathcal{O}_v) \right\rVert_2 \notag \\
    &\;\; + \int\limits_0^t \mathrm{d}s  \left\lVert \mathcal{L}_{\mathrm{<NL}} \mathrm{e}^{\mathcal{L}_{<D}s} | \mathcal{O}_v) \right\rVert_2\label{eq:2terms}
\end{align}
where we have defined $\lVert \mathcal{O}\rVert_2^2 = (\mathcal{O}|\mathcal{O})$.  The left hand side bounds the OTOC which controls the scrambling time.

We first bound the top line of (\ref{eq:2terms}): \begin{lem}
Let $Q_D$ denote the set of vertices exactly distance $D$ from $v$: $Q_D= S_D-S_{D-1}$.  Then there exists $0<\mu<\infty$ such that \begin{equation}
    \lVert \mathbb{P}_{Q_D} \mathrm{e}^{\mathcal{L}_{<D}t}|\mathcal{O}_v) \rVert_2 \le \mathrm{e}^{\mu t-D }. \label{eq:LR}
\end{equation}
\end{lem}
\begin{proof}
This is a well-known result \cite{Lieb1972}; the reader should feel free to skip.  Still,  we present an elegant proof of this lemma, of interest to specialists, using quantum walks \cite{lucas2019nonperturbative,tran2020hierarchy}. In one dimension a slightly improved version of this proof leads to stronger bounds on OTOCs than the provably optimal Lieb-Robinson-style bounds of \cite{chen2019operator,wang2019tightening}. The reason our bounds are stronger is that, as emphasized above, Lieb-Robinson bounds address operator norms, whereas our quantum walk approach directly bounds OTOCs or Frobenius norms.

For notational convenience, we denote for the proof of this lemma $\mathrm{e}^{\mathcal{L}_{<D}t}|\mathcal{O}) = |\mathcal{O}(t))$. Define  \begin{equation}
    \mathcal{F} := \sum_{x\in S_D} b^{d_x} \mathbb{P}_x
\end{equation}
where $d_x$ denotes the distance from $v$ to $x$. Observe that \begin{equation}
    \frac{\mathrm{d}}{\mathrm{d}t} (\mathcal{O}_v(t)|\mathcal{F}|\mathcal{O}_v(t)) = (\mathcal{O}_v(t)|[\mathcal{F},\mathcal{L}_{<D}]|\mathcal{O}_v(t))
\end{equation}
and that it is easy to (crudely) bound the right hand side: denoting $\varphi_x(t) := \lVert \mathbb{P}_x |\mathcal{O}_v(t)) \rVert_2$, we find that (using $|K_{ij}^{AB}|\le 1$) \begin{align}
    \frac{\mathrm{d}}{\mathrm{d}t} (\mathcal{O}_v|\mathcal{F}|\mathcal{O}_v) &\le \sum_{\substack{\lbrace x,y\rbrace \in S_D \\ \text{ and } \lbrace x,y\rbrace \in E}} (\mathcal{O}_v|  [b^{d_x}\mathbb{P}_x + b^{d_y}\mathbb{P}_y, \mathcal{L}_{xy}]|\mathcal{O}_v) \notag \\
    &\le 36 \sum_{\substack{\lbrace x,y\rbrace \in S_D \\ \text{ and } \lbrace x,y\rbrace \in E}} \varphi_x\varphi_y \left(b^{d_x} + b^{d_y}\right) \notag \\
    &\le 18 \sum_{\substack{\lbrace x,y\rbrace \in S_D \\ \text{ and } \lbrace x,y\rbrace \in E}}  \left(b^{d_x} + b^{d_y}\right)\left(\varphi_x^2 + \varphi_y^2\right) \notag \\
    &\le (\mathcal{O}_v|\mathcal{F}|\mathcal{O}_v) \times 18(1+b)k.
\end{align}
In the first line we have defined $\mathcal{L}_{xy} = \mathrm{i}[K^{AB}_{xy}X^A_x X^B_y,\cdot]$; in the third line we have used that $2\varphi_x\varphi_y \le \varphi_x^2+\varphi_y^2$; in the fourth line we have used that only nearest neighbor interactions on $\Lambda$ are allowed by the local terms: if $\lbrace x,y\rbrace \in E$, $|d_x-d_y|\le 1$.
Therefore letting $\mu=9(1+b)k$, \begin{equation}
    (\mathcal{O}_v(t)|\mathcal{F}|\mathcal{O}_v(t)) \le \mathrm{e}^{2\mu t}. \label{eq:expgrowth}
\end{equation}

The final observation is that \begin{equation}
    (\mathcal{O}_v|\mathbb{P}_{Q_D}|\mathcal{O}_v) b^{D} \le \sum_{x\in Q_D} (\mathcal{O}_v|\mathbb{P}_{x}|\mathcal{O}_v)b^D \le (\mathcal{O}_v|\mathcal{F}|\mathcal{O}_v). \label{eq:markov}
\end{equation}
Combining (\ref{eq:expgrowth}) and (\ref{eq:markov}) and setting $b=\mathrm{e}^2$, we obtain (\ref{eq:LR}).
 \end{proof}

This lemma then allows us to crudely (but easily!) bound the first line of (\ref{eq:2terms}) as follows: \begin{align}
    \int\limits_0^t \mathrm{d}s  \left\lVert \mathcal{L}_D \mathrm{e}^{\mathcal{L}_{<D}s} | \mathcal{O}_v) \right\rVert_2 &\le  2t \lVert H_D\rVert \sup_{s\in [0,t]} \left\lVert \mathbb{P}_{Q_D} \mathrm{e}^{\mathcal{L}_{<D}s}|\mathcal{O}_v)\right\rVert_2 \notag \\
    &\le Mt D^{d-1} \mathrm{e}^{\mu t - D}.  \label{eq:last1}
\end{align}
where $\lVert H_D\rVert$ denotes the conventional operator norm (in this case, maximal eigenvalue) of $H_D$ and $M$ is an order 1 constant related to the degree of $\Lambda$.
Then, the second line of (\ref{eq:2terms}) is bounded simply: \begin{align}
     \int\limits_0^t \mathrm{d}s  \left\lVert \mathcal{L}_{\mathrm{<NL}} \mathrm{e}^{\mathcal{L}_{<D}s} | \mathcal{O}_v) \right\rVert_2 &\le t \sup_{s\in[0,t]}\lVert [H_{<\mathrm{NL}},\mathcal{O}_v(s)]\rVert_2 \notag \\  &\le 2t \lVert H_{\mathrm{<NL}}\rVert_2   \label{eq:last2}
\end{align}
where we have used the fact that $\mathcal{O}_v(s)$ has maximal eigenvalue 1 to simplify the calculation above.  Then we observe that \begin{align}
    \lVert H_{\mathrm{<NL}}\rVert_2^2 &= \left(\sum_{i\in S_D, j\notin S_D} + \sum_{\lbrace i,j\rbrace \in S_D}\right) \frac{(J_{ij}^{AB})^2}{N^{2\alpha}} \lVert X^A_i X^B_j \rVert_2^2 \notag \\
    &\le 9|S_D|N^{1-2\alpha}.  \label{eq:last3}
\end{align}
Now let us combine (\ref{eq:last1}), (\ref{eq:last2}) and (\ref{eq:last3}), evaluated at a value of $D$ obeying \begin{equation}
    D \ge 2 \mu t + D_0. \label{eq:D0}
\end{equation}where $D_0$ will be chosen below.  Using (\ref{eq:SDbound}), and when $t$ is large, we conclude  \begin{align}
    \left\lVert\mathbb{P}_{\bar S_D} |\mathcal{O}_v(t))\right\rVert_2 &\le \frac{Z^\prime
    \sqrt{c_1 \max(2\mu t, D_0)^d }t}{ N^{\alpha-1/2}} \notag \\
    &\;\;+ \frac{M^\prime}{2\mu} (2\mu t + D_0)^d \mathrm{e}^{-\mu t-D_0}  \label{eq:last4}
\end{align}
for finite constants $Z^\prime$ and $M^\prime$ independent of $N$. 


We now choose $D_0$ such that \begin{equation}
    \sqrt{\frac{a}{8}} >  \mathrm{e}^{-D_0/2} \sup_{t\in\mathbb{R}^+} \frac{M^\prime}{2\mu} (2\mu t + D_0)^d \mathrm{e}^{-\mu t-D_0/2} \label{eq:sqrta2}
\end{equation}
Note that $D_0$ can be chosen independent of $N$.
To understand why we make this choice, we return to our definition of scrambling time.  Suppose that we choose \begin{equation}
    |S_D|<\frac{aN}{2}. \label{eq:DaN2}
\end{equation}
At the scrambling time $t=t_{\mathrm{s}}$, by (\ref{eq:tsdef}) and (\ref{eq:DaN2}), \begin{equation}
    \frac{a}{2} < (\mathcal{O}_v(t_{\mathrm{s}})|\mathbb{P}_{\bar S_D} |\mathcal{O}_v(t_{\mathrm{s}})). \label{eq:sqrta}
\end{equation}
Now, let us assume that at the scrambling time, $t / N^{1/d} \rightarrow 0$.  In this case, we can always choose a $D$ compatible with (\ref{eq:D0}) and (\ref{eq:DaN2}). Combining (\ref{eq:last4}), (\ref{eq:sqrta2}), and (\ref{eq:sqrta}), we obtain \begin{equation}
    t_{\mathrm{s}}^{1+d/2} > \sqrt{\frac{a}{8}} \frac{N^{\alpha-\frac{1}{2}}}{Z^\prime \sqrt{c_1(2\mu)^d}}, \label{eq:finalthm}
\end{equation}
which leads to (\ref{eq:thm}) so long as $\alpha<1+\frac{1}{d}$.  If instead $\alpha\ge1+\frac{1}{d}$, (\ref{eq:finalthm}) implies that $t_{\mathrm{s}}$ scales faster $N^{1/d}$, which violates our assumption that we could choose a $D$ such that $S_D\subset V$ while (\ref{eq:finalthm}) holds.  \hfill \qedsymbol

\begin{cor}
If the assumptions of Theorem 1 hold, but in addition $K^{AB}_{ij}(t)=0$, then for some $0<C<\infty$, \begin{equation}
    t_{\mathrm{s}}(N) \ge C N^{\alpha-1/2}. \label{eq:cor3}
\end{equation}\label{cor3}
\end{cor}
\begin{proof}
This is a simple extension of the proof of the main theorem.  If $K^{AB}_{ij}=0$, then in (\ref{eq:2terms}) we may consider $D=0$ (i.e. $S_D$ contains only the starting vertex $v$).  (\ref{eq:last4}) reduces to its first term with $|S_D|=1$.  This implies (\ref{eq:cor3}). 
\end{proof}

\section{Tightness of bounds}
We conclude by showing that the simplest of our bounds, Corollary \ref{cor3}, cannot be algebraically improved.  While this is probably not surprising due to the existence of other fast scramblers such as \cite{belyansky2020minimal}, we present an illustrative and simple (not many-body chaotic) protocol that saturates (\ref{eq:cor3}).

\begin{cor}
For $N$ sufficiently large, there exists a time-dependent $H(t)$, satisfying the assumptions of Corollary \ref{cor3} with $\alpha>\frac{1}{2}$, along with \begin{equation}
    J_{ij}^{AB}(t) = J_{ij}(t) \delta^{AZ}\delta^{BZ}
\end{equation} such that for any $\epsilon>0$, \begin{equation}
    t_{\mathrm{s}}(N) \le K N^{\alpha + \epsilon - 1/2} \label{eq:cor4}
\end{equation}
where $K$ is a finite constant that can depend on $\epsilon$.  
\end{cor}
\begin{proof}
We describe a simple protocol to grow large operators. We choose a $g=\mathrm{O}(1)$ and $M=\mathrm{O}(N)$ (to be specified more carefully later) such that \begin{equation}
    N \ge 1 + gM. \label{eq:NgM}
\end{equation}
Let $R_1,\ldots,R_g$ denote $g$ disjoint sets of vertices with $M$ elements.  By (\ref{eq:NgM}) there exists another vertex (let's call it 0) not in any of these sets.

For simplicity, we work in a basis of Pauli matrices $\lbrace 1, X^+, X^-, Z\rbrace$ on every site, where we define $\sqrt{2}X^\pm = X\pm \mathrm{i}Y$ (note the slightly unusual normalization).  Our goal is to build a protocol that starts with $X^+_0$ and time evolves it into an operator with average size of order $M$.  The protocol will work by first expanding the operator into set $R_1$ using two-body $ZZ$ couplings, then applying a rotation to convert all $Z$ in $R_1$ into $X$, then expanding into $R_2$ using $ZZ$ couplings, and so on.  After $l$ steps, the size of the operator will scale as $C^l$ with high probability. To be precise, we say that an operator $\mathcal{O}$ has size $s$ with probability $P_s$, where \begin{equation}
    P_s = (\mathcal{O}|\mathbb{Q}_s|\mathcal{O})
\end{equation} 
where $\mathbb{Q}_s$ is a projection superoperator onto Pauli strings with exactly $s$ non-identity components.  We will then show how to choose $g$ so that $C^g=M$ and (\ref{eq:cor4}) are both obeyed.

Let us now show how to achieve the goals outlined above, starting with the first step of the protocol.  Let \begin{equation}
    U_{Z,1} = \mathrm{e}^{-\mathrm{i} \frac{\tau}{2} Z_0 Z_{R_1}}
\end{equation}
where we have defined \begin{equation}
    X^A_{R_i} := \sum_{v\in R_i} X^A_v.
\end{equation}
A straightforward calculation (see e.g. \cite{lucas2019star}) shows that \begin{equation}
   U_{Z,1}^\dagger X^+_0 U_{Z,1} = X^+_0 \mathrm{e}^{\mathrm{i}\tau Z_{R_1}}. \label{eq:UZX0}
\end{equation}
The probability that $U_{Z,1}^\dagger X^+_0 U_{Z,1}$ has size $1\le s\le M+1$ after step 1 of the protocol is then \begin{equation}
    P_{s,1} = \left(\begin{array}{c} M \\ s-1 \end{array}\right) \left(\cos^2\tau\right)^{M-s-1} \left(\sin^2 \tau\right)^{s+1}.
\end{equation}
The average size of the operator is \begin{equation}
    \sum sP_{s,1} = 1 + M\sin^2\tau .
\end{equation}
Since the distribution is binomial, fluctuations about the mean are of order $\sqrt{M}|\sin\tau|$ and if $\tau$ is sufficiently large, these fluctuations will be small. Let us define an O(1) constants $c_{1,2}$ obeying $c_1<1<c_2$ such that \begin{equation}
    p_0 := 2^{-1/2g} < \sum_{s= 1 + s_1^* }^{1+\lceil c_2s_1^* \rceil} P_{s,1}.
\end{equation} 
where \begin{equation} 
    s_1^* := \lceil c_1 M\sin^2\tau \rceil.
\end{equation}
Most of the operator has size at least $s_1^*+1$.  So, if we define the projection \begin{equation}
    \mathbb{R}_1 := \sum_{s= 1 + \lceil c_1 M\sin^2\tau \rceil }^{1+\lceil c_2s_1^* \rceil} \mathbb{Q}_s,
\end{equation}  
then it suffices to keep track of only the operator $\mathbb{R}_1 U^\dagger_Z X^+_0 U_{Z,1}$ . Note also that there exists a finite constant $c_3$ such that the time it takes to run the unitary $U_{Z,1}$ by turning on only Ising couplings as given in (\ref{eq:thmH}) is given by \begin{equation}
    t_Z = c_3 N^{\alpha}\tau.
\end{equation}

The next step of the protocol corresponds to rotating the $Z$s in (\ref{eq:UZX0}) into $X = (X^+ + X^-)/\sqrt{2}$.  We can do this in an O(1) time $t_X$ using $H(t)$ in (\ref{eq:thmH}) by applying the unitary transformation \begin{equation}
    U_{Y,1} = \mathrm{e}^{-\mathrm{i}\frac{\pi}{4}Y_{R_1}}.
\end{equation}  
Define the final operator to be \begin{equation}
\mathcal{O}_1 = U^\dagger_{Y,1}U^\dagger_{Z,1}X^+_0U_{Z,1}U_{Y,1} 
\end{equation}
After this step, the probability that a Pauli string has $l_\pm$ $X^\pm$ Pauli strings in $R_1$ is given as follows.  Let\begin{equation}
    j = l_+ - l_-
\end{equation}
denote the difference between the number of $X^+$ and $X^-$, and let $\mathbb{J}^1_j$ project on to Pauli strings with this imbalance $j$.  Then \begin{equation}
    (\mathcal{O}_1| \mathbb{Q}_s \mathbb{J}^1_j \mathbb{Q}_s|\mathcal{O}_1) =  (\mathcal{O}_1| \mathbb{Q}_s |\mathcal{O}_1) \left(\begin{array}{c} s \\ \frac{1}{2}(s+j) \end{array}\right) \frac{1}{2^s} \label{eq:jbinomial}
\end{equation}
We define an O(1) constant $c_4$ such that for any value of $s\ge \lceil c_1 M\sin^2\tau \rceil$, \begin{equation}
    p_0 < \sum_{j : |j|>c_4 \sqrt{s}} \left(\begin{array}{c} s \\ \frac{1}{2}(s+j) \end{array}\right) \frac{1}{2^s}.
\end{equation}
Define the projector $\mathbb{K}_1$ onto all Pauli strings with $s$ $X^\pm$ in $R_1$, such that $s>\lceil c_1 M\sin^2\tau \rceil$ and with $|j|>c_4 \sqrt{s}$.  Then clearly, \begin{equation}
    \lVert \mathbb{K}_1 |\mathcal{O}_1) \rVert_2 \ge p_0^2.
\end{equation}

For the remaining steps of the protocol, we choose the unitaries \begin{subequations}
    \begin{align}
        U_{Z,l} &= \mathrm{e}^{-\mathrm{i}\frac{\tau}{2} Z_{R_{l-1}}Z_{R_l}}, \\
        U_{Y,l} &= \mathrm{e}^{-\mathrm{i}\frac{\pi}{4} Y_{R_l}}.
    \end{align}
\end{subequations}
We define \begin{equation}
    \mathcal{O}_l := U^\dagger_{Y,l}U^\dagger_{Z,l}\mathcal{O}_{l-1}U_{Z,l}U_{Y,l}.
\end{equation}
The key observation is that $\mathcal{O}_l$ consists entirely of Pauli strings of identity, $X^+$ and $X^-$, \emph{and} that every single Pauli string in $\mathcal{O}_{l-1}$, in our $\pm$ basis, evolves in an orthogonal subspace of operator space $\mathcal{B}$ relative to every other string during steps $l,l+1,\ldots g$.  Therefore, we can easily recursively analyze the operator growth after every step of the protocol. 

For example, when $l=2$, we can analyze the evolution of the operator $\mathbb{J}^1_j\mathbb{Q}_s \mathcal{O}_1$ separately for each $j$ and each $s$.  Upon doing so, we find that the probability for having size $s_2$ in domain $R_2$, with imbalance $j_2$ between $X^+$ and $X^-$ in $R_2$, is given by \begin{equation}
    P_{s_2,2}(j) = \left(\begin{array}{c} M \\ s_2 \end{array}\right) \left(\cos^2(j\tau)\right)^{M-s_2} \left(\sin^2 (j\tau)\right)^{s_2}
\end{equation}
    The answer only depends on $j$, since for any operator $A$
    \begin{equation}
        [Z_{R_1}, \mathbb{J}^1_j A] = 2\mathrm{i}j\mathbb{J}^1_j A.
    \end{equation}
    Now, we define the projector $\mathbb{R}_2$ onto all operators with at least $s_2^*$ $Z$s in $R_2$ (and not more than $c_2 s_2^*$), where \begin{equation}
        s_2^* := M\sin^2\left(c_4 \sqrt{s_1^*} \tau\right) > \frac{4c_4^2 }{\pi^2} \left(s_1^*\right)^2. 
    \end{equation}
    In deriving this formula, it was important that we could ignore zeros of $\sin^2(j\tau)$ away from $\tau=0$; we will confirm at the end of the proof that this is so.
    The distribution of $j_2$ after applying $U_{Y,2}$ is given by the binomial formula, similar to (\ref{eq:jbinomial}).    A straightforward generalization of the logic at step 1 tells us that \begin{equation}
        \lVert \mathbb{K}_2|\mathcal{O}_2)\rVert_2 \ge p_0^4
    \end{equation}
    where $\mathbb{K}_2$ projects onto all Pauli strings with at least $s_2^*$ $X^\pm$ in $R_2$ and with imbalance of at least $c_4 \sqrt{s_2^*}$.

    Clearly this procedure extends to all $l$.  The minimal size after each step is \begin{equation}
        s_l^* > \left(\frac{4c_4^2 }{\pi^2}\right)^{l-1} \left(s_1^*\right)^l.
    \end{equation}  After $g$ steps, we conclude that \begin{align}
        \sum_{i=1}^N (\mathcal{O}_g|\mathbb{P}_i|\mathcal{O}_g) &\ge \sum_{i=1}^N (\mathcal{O}_g|\mathbb{K}_g \mathbb{P}_i \mathbb{K}_g |\mathcal{O}_g) \ge s_g^* (\mathcal{O}_g | \mathbb{K}_g|\mathcal{O}_g)  \notag \\
        &\ge s_g^* p_0^{2g}=  \frac{s_g^*}{2}.
    \end{align}
    The total runtime of the protocol is \begin{equation}
        t_{\mathrm{s}} = g(t_Z+t_X).
    \end{equation}
    
    It remains to fix $g$ and $\tau$ such that (\ref{eq:cor4}) holds.  Since $t_X$ is O(1) it suffices to choose $g$ and $\tau$ such that for some $0<c_5<\infty$ \begin{subequations}\label{eq:cor4end}\begin{align}
        gc_3N^\alpha\tau &\le KN^{\alpha+\epsilon-1/2}, \\
        c_5 M &\le s_g^*.
        \end{align}
    \end{subequations}
    These inequalities hold if\begin{equation}
        s_1^* = \frac{\pi^2}{4c_4^2} \left(\frac{4c_4^2c_5N}{\pi^2 g}\right)^{1/g} > \frac{4}{\pi^2} c_1\frac{N}{g}\tau^2 \label{eq:s1starend}
    \end{equation}
    or, for a suitable $0<c_6<\infty$ that depends on $g$, \begin{equation}
        \tau = c_6 N^{-(g-1)/2g}. \label{eq:tauchoice}
    \end{equation}
        Choosing $g>1/2\epsilon$, we satisfy (\ref{eq:cor4end}).
        
        The final thing to confirm is that the ``imbalance" of $X^+$ and $X^-$ is always so small that $j\tau \ll 1$ holds, except at the final step.  Since $j\sim \sqrt{s}$ with high probability, (\ref{eq:s1starend}) and (\ref{eq:tauchoice}) confirm that this is so.  In the final step, due to the binomial distribution of $j$, there will be negligible concentration around $j\tau/\pi \in \mathbb{Z}$, and the probability that the operator has size $\mathrm{O}(M)$ is finite.
\end{proof}

\section*{Acknowledgements}
We thank Alexey Gorshkov, Ana Maria Rey and Brian Swingle for useful discussions.

\bibliography{scram_all2all_spin}

\providecommand{\noopsort}[1]{}\providecommand{\singleletter}[1]{#1}%
\begin{thebibliography}{36}%
\makeatletter
\providecommand \@ifxundefined [1]{%
 \@ifx{#1\undefined}
}%
\providecommand \@ifnum [1]{%
 \ifnum #1\expandafter \@firstoftwo
 \else \expandafter \@secondoftwo
 \fi
}%
\providecommand \@ifx [1]{%
 \ifx #1\expandafter \@firstoftwo
 \else \expandafter \@secondoftwo
 \fi
}%
\providecommand \natexlab [1]{#1}%
\providecommand \enquote  [1]{``#1''}%
\providecommand \bibnamefont  [1]{#1}%
\providecommand \bibfnamefont [1]{#1}%
\providecommand \citenamefont [1]{#1}%
\providecommand \href@noop [0]{\@secondoftwo}%
\providecommand \href [0]{\begingroup \@sanitize@url \@href}%
\providecommand \@href[1]{\@@startlink{#1}\@@href}%
\providecommand \@@href[1]{\endgroup#1\@@endlink}%
\providecommand \@sanitize@url [0]{\catcode `\\12\catcode `\$12\catcode
  `\&12\catcode `\#12\catcode `\^12\catcode `\_12\catcode `\%12\relax}%
\providecommand \@@startlink[1]{}%
\providecommand \@@endlink[0]{}%
\providecommand \url  [0]{\begingroup\@sanitize@url \@url }%
\providecommand \@url [1]{\endgroup\@href {#1}{\urlprefix }}%
\providecommand \urlprefix  [0]{URL }%
\providecommand \Eprint [0]{\href }%
\providecommand \doibase [0]{http://dx.doi.org/}%
\providecommand \selectlanguage [0]{\@gobble}%
\providecommand \bibinfo  [0]{\@secondoftwo}%
\providecommand \bibfield  [0]{\@secondoftwo}%
\providecommand \translation [1]{[#1]}%
\providecommand \BibitemOpen [0]{}%
\providecommand \bibitemStop [0]{}%
\providecommand \bibitemNoStop [0]{.\EOS\space}%
\providecommand \EOS [0]{\spacefactor3000\relax}%
\providecommand \BibitemShut  [1]{\csname bibitem#1\endcsname}%
\let\auto@bib@innerbib\@empty
\bibitem [{\citenamefont {Maldacena}(1999)}]{Maldacena:1997re}%
  \BibitemOpen
  \bibfield  {author} {\bibinfo {author} {\bibfnamefont {Juan~Martin}\
  \bibnamefont {Maldacena}},\ }\bibfield  {title} {\enquote {\bibinfo {title}
  {{The Large N limit of superconformal field theories and supergravity}},}\
  }\href {\doibase 10.1023/A:1026654312961} {\bibfield  {journal} {\bibinfo
  {journal} {Int. J. Theor. Phys.}\ }\textbf {\bibinfo {volume} {38}},\
  \bibinfo {pages} {1113--1133} (\bibinfo {year} {1999})},\ \Eprint
  {http://arxiv.org/abs/hep-th/9711200} {arXiv:hep-th/9711200} \BibitemShut
  {NoStop}%
\bibitem [{\citenamefont {Sachdev}(2015)}]{Sachdev:2015efa}%
  \BibitemOpen
  \bibfield  {author} {\bibinfo {author} {\bibfnamefont {Subir}\ \bibnamefont
  {Sachdev}},\ }\bibfield  {title} {\enquote {\bibinfo {title}
  {{{Bekenstein-Hawking} Entropy and Strange Metals}},}\ }\href {\doibase
  10.1103/PhysRevX.5.041025} {\bibfield  {journal} {\bibinfo  {journal} {Phys.
  Rev. X}\ }\textbf {\bibinfo {volume} {5}},\ \bibinfo {pages} {041025}
  (\bibinfo {year} {2015})}\BibitemShut {NoStop}%
\bibitem [{\citenamefont {Maldacena}\ and\ \citenamefont
  {Stanford}(2016)}]{Maldacena:2016hyu}%
  \BibitemOpen
  \bibfield  {author} {\bibinfo {author} {\bibfnamefont {Juan}\ \bibnamefont
  {Maldacena}}\ and\ \bibinfo {author} {\bibfnamefont {Douglas}\ \bibnamefont
  {Stanford}},\ }\bibfield  {title} {\enquote {\bibinfo {title} {{Remarks on
  the {Sachdev-Ye-Kitaev} model}},}\ }\href {\doibase
  10.1103/PhysRevD.94.106002} {\bibfield  {journal} {\bibinfo  {journal} {Phys.
  Rev. D}\ }\textbf {\bibinfo {volume} {94}},\ \bibinfo {pages} {106002}
  (\bibinfo {year} {2016})}\BibitemShut {NoStop}%
\bibitem [{\citenamefont {Kitaev}\ and\ \citenamefont
  {Suh}(2018)}]{Kitaev:2017awl}%
  \BibitemOpen
  \bibfield  {author} {\bibinfo {author} {\bibfnamefont {Alexei}\ \bibnamefont
  {Kitaev}}\ and\ \bibinfo {author} {\bibfnamefont {S.~Josephine}\ \bibnamefont
  {Suh}},\ }\bibfield  {title} {\enquote {\bibinfo {title} {{The soft mode in
  the Sachdev-Ye-Kitaev model and its gravity dual}},}\ }\href {\doibase
  10.1007/JHEP05(2018)183} {\bibfield  {journal} {\bibinfo  {journal} {JHEP}\
  }\textbf {\bibinfo {volume} {05}},\ \bibinfo {pages} {183} (\bibinfo {year}
  {2018})}\BibitemShut {NoStop}%
\bibitem [{\citenamefont {Chew}\ \emph {et~al.}(2017)\citenamefont {Chew},
  \citenamefont {Essin},\ and\ \citenamefont {Alicea}}]{Chew_2017}%
  \BibitemOpen
  \bibfield  {author} {\bibinfo {author} {\bibfnamefont {Aaron}\ \bibnamefont
  {Chew}}, \bibinfo {author} {\bibfnamefont {Andrew}\ \bibnamefont {Essin}}, \
  and\ \bibinfo {author} {\bibfnamefont {Jason}\ \bibnamefont {Alicea}},\
  }\bibfield  {title} {\enquote {\bibinfo {title} {Approximating the
  {Sachdev-Ye-Kitaev model with Majorana wires}},}\ }\href
  {http://dx.doi.org/10.1103/PhysRevB.96.121119} {\bibfield  {journal}
  {\bibinfo  {journal} {Physical Review B}\ }\textbf {\bibinfo {volume} {96}},\
  \bibinfo {pages} {121119} (\bibinfo {year} {2017})}\BibitemShut {NoStop}%
\bibitem [{\citenamefont {Chen}\ \emph {et~al.}(2018)\citenamefont {Chen},
  \citenamefont {Ilan}, \citenamefont {de~Juan}, \citenamefont {Pikulin},\ and\
  \citenamefont {Franz}}]{Chen_2018}%
  \BibitemOpen
  \bibfield  {author} {\bibinfo {author} {\bibfnamefont {Anffany}\ \bibnamefont
  {Chen}}, \bibinfo {author} {\bibfnamefont {R.}~\bibnamefont {Ilan}}, \bibinfo
  {author} {\bibfnamefont {F.}~\bibnamefont {de~Juan}}, \bibinfo {author}
  {\bibfnamefont {D.~I.}\ \bibnamefont {Pikulin}}, \ and\ \bibinfo {author}
  {\bibfnamefont {M.}~\bibnamefont {Franz}},\ }\bibfield  {title} {\enquote
  {\bibinfo {title} {Quantum holography in a graphene flake with an irregular
  boundary},}\ }\href {http://dx.doi.org/10.1103/PhysRevLett.121.036403}
  {\bibfield  {journal} {\bibinfo  {journal} {Physical Review Letters}\
  }\textbf {\bibinfo {volume} {121}},\ \bibinfo {pages} {036403} (\bibinfo
  {year} {2018})}\BibitemShut {NoStop}%
\bibitem [{\citenamefont {Marino}\ and\ \citenamefont
  {Rey}(2019)}]{Marino_2019}%
  \BibitemOpen
  \bibfield  {author} {\bibinfo {author} {\bibfnamefont {J.}~\bibnamefont
  {Marino}}\ and\ \bibinfo {author} {\bibfnamefont {A.~M.}\ \bibnamefont
  {Rey}},\ }\bibfield  {title} {\enquote {\bibinfo {title} {{Cavity-QED}
  simulator of slow and fast scrambling},}\ }\href
  {http://dx.doi.org/10.1103/PhysRevA.99.051803} {\bibfield  {journal}
  {\bibinfo  {journal} {Physical Review A}\ }\textbf {\bibinfo {volume} {99}},\
  \bibinfo {pages} {051803} (\bibinfo {year} {2019})}\BibitemShut {NoStop}%
\bibitem [{\citenamefont {Lewis-Swan}\ \emph {et~al.}(2019)\citenamefont
  {Lewis-Swan}, \citenamefont {Safavi-Naini}, \citenamefont {Bollinger},\ and\
  \citenamefont {Rey}}]{Lewis_Swan_2019}%
  \BibitemOpen
  \bibfield  {author} {\bibinfo {author} {\bibfnamefont {R.~J.}\ \bibnamefont
  {Lewis-Swan}}, \bibinfo {author} {\bibfnamefont {A.}~\bibnamefont
  {Safavi-Naini}}, \bibinfo {author} {\bibfnamefont {J.~J.}\ \bibnamefont
  {Bollinger}}, \ and\ \bibinfo {author} {\bibfnamefont {A.~M.}\ \bibnamefont
  {Rey}},\ }\bibfield  {title} {\enquote {\bibinfo {title} {Unifying
  scrambling, thermalization and entanglement through measurement of fidelity
  out-of-time-order correlators in the {Dicke} model},}\ }\href
  {http://dx.doi.org/10.1038/s41467-019-09436-y} {\bibfield  {journal}
  {\bibinfo  {journal} {Nature Communications}\ }\textbf {\bibinfo {volume}
  {10}},\ \bibinfo {pages} {1581} (\bibinfo {year} {2019})}\BibitemShut
  {NoStop}%
\bibitem [{\citenamefont {Alavirad}\ and\ \citenamefont
  {Lavasani}(2019)}]{Alavirad_2019}%
  \BibitemOpen
  \bibfield  {author} {\bibinfo {author} {\bibfnamefont {Yahya}\ \bibnamefont
  {Alavirad}}\ and\ \bibinfo {author} {\bibfnamefont {Ali}\ \bibnamefont
  {Lavasani}},\ }\bibfield  {title} {\enquote {\bibinfo {title} {Scrambling in
  the {Dicke} model},}\ }\href {http://dx.doi.org/10.1103/PhysRevA.99.043602}
  {\bibfield  {journal} {\bibinfo  {journal} {Physical Review A}\ }\textbf
  {\bibinfo {volume} {99}},\ \bibinfo {pages} {043602} (\bibinfo {year}
  {2019})}\BibitemShut {NoStop}%
\bibitem [{\citenamefont {Bentsen}\ \emph
  {et~al.}(2019{\natexlab{a}})\citenamefont {Bentsen}, \citenamefont
  {Hashizume}, \citenamefont {Buyskikh}, \citenamefont {Davis}, \citenamefont
  {Daley}, \citenamefont {Gubser},\ and\ \citenamefont
  {Schleier-Smith}}]{Bentsen_2019prl}%
  \BibitemOpen
  \bibfield  {author} {\bibinfo {author} {\bibfnamefont {Gregory}\ \bibnamefont
  {Bentsen}}, \bibinfo {author} {\bibfnamefont {Tomohiro}\ \bibnamefont
  {Hashizume}}, \bibinfo {author} {\bibfnamefont {Anton~S.}\ \bibnamefont
  {Buyskikh}}, \bibinfo {author} {\bibfnamefont {Emily~J.}\ \bibnamefont
  {Davis}}, \bibinfo {author} {\bibfnamefont {Andrew~J.}\ \bibnamefont
  {Daley}}, \bibinfo {author} {\bibfnamefont {Steven~S.}\ \bibnamefont
  {Gubser}}, \ and\ \bibinfo {author} {\bibfnamefont {Monika}\ \bibnamefont
  {Schleier-Smith}},\ }\bibfield  {title} {\enquote {\bibinfo {title} {Treelike
  interactions and fast scrambling with cold atoms},}\ }\href
  {http://dx.doi.org/10.1103/PhysRevLett.123.130601} {\bibfield  {journal}
  {\bibinfo  {journal} {Physical Review Letters}\ }\textbf {\bibinfo {volume}
  {123}},\ \bibinfo {pages} {130601} (\bibinfo {year}
  {2019}{\natexlab{a}})}\BibitemShut {NoStop}%
\bibitem [{\citenamefont {Bentsen}\ \emph
  {et~al.}(2019{\natexlab{b}})\citenamefont {Bentsen}, \citenamefont
  {Potirniche}, \citenamefont {Bulchandani}, \citenamefont {Scaffidi},
  \citenamefont {Cao}, \citenamefont {Qi}, \citenamefont {Schleier-Smith},\
  and\ \citenamefont {Altman}}]{Bentsen_2019prx}%
  \BibitemOpen
  \bibfield  {author} {\bibinfo {author} {\bibfnamefont {Gregory}\ \bibnamefont
  {Bentsen}}, \bibinfo {author} {\bibfnamefont {Ionut-Dragos}\ \bibnamefont
  {Potirniche}}, \bibinfo {author} {\bibfnamefont {Vir~B.}\ \bibnamefont
  {Bulchandani}}, \bibinfo {author} {\bibfnamefont {Thomas}\ \bibnamefont
  {Scaffidi}}, \bibinfo {author} {\bibfnamefont {Xiangyu}\ \bibnamefont {Cao}},
  \bibinfo {author} {\bibfnamefont {Xiao-Liang}\ \bibnamefont {Qi}}, \bibinfo
  {author} {\bibfnamefont {Monika}\ \bibnamefont {Schleier-Smith}}, \ and\
  \bibinfo {author} {\bibfnamefont {Ehud}\ \bibnamefont {Altman}},\ }\bibfield
  {title} {\enquote {\bibinfo {title} {Integrable and chaotic dynamics of spins
  coupled to an optical cavity},}\ }\href
  {http://dx.doi.org/10.1103/PhysRevX.9.041011} {\bibfield  {journal} {\bibinfo
   {journal} {Physical Review X}\ }\textbf {\bibinfo {volume} {9}},\ \bibinfo
  {pages} {041011} (\bibinfo {year} {2019}{\natexlab{b}})}\BibitemShut
  {NoStop}%
\bibitem [{\citenamefont {Sekino}\ and\ \citenamefont
  {Susskind}(2008)}]{Sekino:2008he}%
  \BibitemOpen
  \bibfield  {author} {\bibinfo {author} {\bibfnamefont {Yasuhiro}\
  \bibnamefont {Sekino}}\ and\ \bibinfo {author} {\bibfnamefont {Leonard}\
  \bibnamefont {Susskind}},\ }\bibfield  {title} {\enquote {\bibinfo {title}
  {{Fast Scramblers}},}\ }\href {\doibase 10.1088/1126-6708/2008/10/065}
  {\bibfield  {journal} {\bibinfo  {journal} {JHEP}\ }\textbf {\bibinfo
  {volume} {10}},\ \bibinfo {pages} {065} (\bibinfo {year} {2008})}\BibitemShut
  {NoStop}%
\bibitem [{\citenamefont {Shenker}\ and\ \citenamefont
  {Stanford}(2014)}]{Shenker:2013pqa}%
  \BibitemOpen
  \bibfield  {author} {\bibinfo {author} {\bibfnamefont {Stephen~H.}\
  \bibnamefont {Shenker}}\ and\ \bibinfo {author} {\bibfnamefont {Douglas}\
  \bibnamefont {Stanford}},\ }\bibfield  {title} {\enquote {\bibinfo {title}
  {{Black holes and the butterfly effect}},}\ }\href {\doibase
  10.1007/JHEP03(2014)067} {\bibfield  {journal} {\bibinfo  {journal} {JHEP}\
  }\textbf {\bibinfo {volume} {03}},\ \bibinfo {pages} {067} (\bibinfo {year}
  {2014})}\BibitemShut {NoStop}%
\bibitem [{\citenamefont {Lieb}\ and\ \citenamefont
  {Robinson}(1972)}]{Lieb1972}%
  \BibitemOpen
  \bibfield  {author} {\bibinfo {author} {\bibfnamefont {Elliott~H.}\
  \bibnamefont {Lieb}}\ and\ \bibinfo {author} {\bibfnamefont {Derek~W.}\
  \bibnamefont {Robinson}},\ }\bibfield  {title} {\enquote {\bibinfo {title}
  {The finite group velocity of quantum spin systems},}\ }\href {\doibase
  10.1007/BF01645779} {\bibfield  {journal} {\bibinfo  {journal} {Commun. Math.
  Phys.}\ }\textbf {\bibinfo {volume} {28}},\ \bibinfo {pages} {251--257}
  (\bibinfo {year} {1972})}\BibitemShut {NoStop}%
\bibitem [{\citenamefont {Lashkari}\ \emph {et~al.}(2013)\citenamefont
  {Lashkari}, \citenamefont {Stanford}, \citenamefont {Hastings}, \citenamefont
  {Osborne},\ and\ \citenamefont {Hayden}}]{Lashkari_2013}%
  \BibitemOpen
  \bibfield  {author} {\bibinfo {author} {\bibfnamefont {Nima}\ \bibnamefont
  {Lashkari}}, \bibinfo {author} {\bibfnamefont {Douglas}\ \bibnamefont
  {Stanford}}, \bibinfo {author} {\bibfnamefont {Matthew}\ \bibnamefont
  {Hastings}}, \bibinfo {author} {\bibfnamefont {Tobias}\ \bibnamefont
  {Osborne}}, \ and\ \bibinfo {author} {\bibfnamefont {Patrick}\ \bibnamefont
  {Hayden}},\ }\bibfield  {title} {\enquote {\bibinfo {title} {{Towards the
  Fast Scrambling Conjecture}},}\ }\href {\doibase 10.1007/JHEP04(2013)022}
  {\bibfield  {journal} {\bibinfo  {journal} {JHEP}\ }\textbf {\bibinfo
  {volume} {04}},\ \bibinfo {pages} {022} (\bibinfo {year} {2013})},\ \Eprint
  {http://arxiv.org/abs/1111.6580} {arXiv:1111.6580 [hep-th]} \BibitemShut
  {NoStop}%
\bibitem [{\citenamefont {Bentsen}\ \emph
  {et~al.}(2019{\natexlab{c}})\citenamefont {Bentsen}, \citenamefont {Gu},\
  and\ \citenamefont {Lucas}}]{Bentsen_2019}%
  \BibitemOpen
  \bibfield  {author} {\bibinfo {author} {\bibfnamefont {Gregory}\ \bibnamefont
  {Bentsen}}, \bibinfo {author} {\bibfnamefont {Yingfei}\ \bibnamefont {Gu}}, \
  and\ \bibinfo {author} {\bibfnamefont {Andrew}\ \bibnamefont {Lucas}},\
  }\bibfield  {title} {\enquote {\bibinfo {title} {Fast scrambling on sparse
  graphs},}\ }\href {\doibase 10.1073/pnas.1811033116} {\bibfield  {journal}
  {\bibinfo  {journal} {Proceedings of the National Academy of Sciences}\
  }\textbf {\bibinfo {volume} {116}},\ \bibinfo {pages} {6689–6694} (\bibinfo
  {year} {2019}{\natexlab{c}})}\BibitemShut {NoStop}%
\bibitem [{\citenamefont {Guo}\ \emph {et~al.}(2019)\citenamefont {Guo},
  \citenamefont {Tran}, \citenamefont {Childs}, \citenamefont {Gorshkov},\ and\
  \citenamefont {Gong}}]{guo2019signaling}%
  \BibitemOpen
  \bibfield  {author} {\bibinfo {author} {\bibfnamefont {Andrew~Y.}\
  \bibnamefont {Guo}}, \bibinfo {author} {\bibfnamefont {Minh~C.}\ \bibnamefont
  {Tran}}, \bibinfo {author} {\bibfnamefont {Andrew~M.}\ \bibnamefont
  {Childs}}, \bibinfo {author} {\bibfnamefont {Alexey~V.}\ \bibnamefont
  {Gorshkov}}, \ and\ \bibinfo {author} {\bibfnamefont {Zhe-Xuan}\ \bibnamefont
  {Gong}},\ }\href@noop {} {\enquote {\bibinfo {title} {Signaling and
  scrambling with strongly long-range interactions},}\ } (\bibinfo {year}
  {2019}),\ \Eprint {http://arxiv.org/abs/1906.02662} {arXiv:1906.02662
  [quant-ph]} \BibitemShut {NoStop}%
\bibitem [{\citenamefont {Britton}\ \emph {et~al.}(2012)\citenamefont
  {Britton}, \citenamefont {Sawyer}, \citenamefont {Keith}, \citenamefont
  {Wang}, \citenamefont {Freericks}, \citenamefont {Uys}, \citenamefont
  {Biercuk},\ and\ \citenamefont {Bollinger}}]{Britton_2012}%
  \BibitemOpen
  \bibfield  {author} {\bibinfo {author} {\bibfnamefont {Joseph~W.}\
  \bibnamefont {Britton}}, \bibinfo {author} {\bibfnamefont {Brian~C.}\
  \bibnamefont {Sawyer}}, \bibinfo {author} {\bibfnamefont {Adam~C.}\
  \bibnamefont {Keith}}, \bibinfo {author} {\bibfnamefont {C.-C.~Joseph}\
  \bibnamefont {Wang}}, \bibinfo {author} {\bibfnamefont {James~K.}\
  \bibnamefont {Freericks}}, \bibinfo {author} {\bibfnamefont {Hermann}\
  \bibnamefont {Uys}}, \bibinfo {author} {\bibfnamefont {Michael~J.}\
  \bibnamefont {Biercuk}}, \ and\ \bibinfo {author} {\bibfnamefont {John~J.}\
  \bibnamefont {Bollinger}},\ }\bibfield  {title} {\enquote {\bibinfo {title}
  {{Engineered two-dimensional Ising interactions in a trapped-ion quantum
  simulator with hundreds of spins}},}\ }\href
  {http://dx.doi.org/10.1038/nature10981} {\bibfield  {journal} {\bibinfo
  {journal} {Nature}\ }\textbf {\bibinfo {volume} {484}},\ \bibinfo {pages}
  {489–492} (\bibinfo {year} {2012})}\BibitemShut {NoStop}%
\bibitem [{\citenamefont {Leroux}\ \emph {et~al.}(2010)\citenamefont {Leroux},
  \citenamefont {Schleier-Smith},\ and\ \citenamefont
  {Vuletić}}]{Leroux_2010}%
  \BibitemOpen
  \bibfield  {author} {\bibinfo {author} {\bibfnamefont {Ian~D.}\ \bibnamefont
  {Leroux}}, \bibinfo {author} {\bibfnamefont {Monika~H.}\ \bibnamefont
  {Schleier-Smith}}, \ and\ \bibinfo {author} {\bibfnamefont {Vladan}\
  \bibnamefont {Vuletić}},\ }\bibfield  {title} {\enquote {\bibinfo {title}
  {Implementation of cavity squeezing of a collective atomic spin},}\ }\href
  {http://dx.doi.org/10.1103/PhysRevLett.104.073602} {\bibfield  {journal}
  {\bibinfo  {journal} {Physical Review Letters}\ }\textbf {\bibinfo {volume}
  {104}},\ \bibinfo {pages} {073602} (\bibinfo {year} {2010})}\BibitemShut
  {NoStop}%
\bibitem [{\citenamefont {Muniz}\ \emph {et~al.}(2020)\citenamefont {Muniz},
  \citenamefont {Barbarena}, \citenamefont {Lewis-Swan}, \citenamefont {Young},
  \citenamefont {Cline}, \citenamefont {Rey},\ and\ \citenamefont
  {Thompson}}]{thompson2020}%
  \BibitemOpen
  \bibfield  {author} {\bibinfo {author} {\bibfnamefont {J.~A.}\ \bibnamefont
  {Muniz}}, \bibinfo {author} {\bibfnamefont {D.}~\bibnamefont {Barbarena}},
  \bibinfo {author} {\bibfnamefont {R.~J.}\ \bibnamefont {Lewis-Swan}},
  \bibinfo {author} {\bibfnamefont {D.~J.}\ \bibnamefont {Young}}, \bibinfo
  {author} {\bibfnamefont {J.~R.~K.}\ \bibnamefont {Cline}}, \bibinfo {author}
  {\bibfnamefont {A.~M.}\ \bibnamefont {Rey}}, \ and\ \bibinfo {author}
  {\bibfnamefont {J.~K.}\ \bibnamefont {Thompson}},\ }\bibfield  {title}
  {\enquote {\bibinfo {title} {Exploring dynamical phase transitions with cold
  atoms in an optical cavity},}\ }\href
  {https://doi.org/10.1038/s41586-020-2224-x} {\bibfield  {journal} {\bibinfo
  {journal} {Nature}\ }\textbf {\bibinfo {volume} {580}},\ \bibinfo {pages}
  {602--607} (\bibinfo {year} {2020})}\BibitemShut {NoStop}%
\bibitem [{\citenamefont {Garttner}\ \emph {et~al.}(2017)\citenamefont
  {Garttner}, \citenamefont {Bohnet}, \citenamefont {Safavi-Naini},
  \citenamefont {Wall}, \citenamefont {Bollinger},\ and\ \citenamefont
  {Rey}}]{Garttner_2017}%
  \BibitemOpen
  \bibfield  {author} {\bibinfo {author} {\bibfnamefont {Martin}\ \bibnamefont
  {Garttner}}, \bibinfo {author} {\bibfnamefont {Justin~G.}\ \bibnamefont
  {Bohnet}}, \bibinfo {author} {\bibfnamefont {Arghavan}\ \bibnamefont
  {Safavi-Naini}}, \bibinfo {author} {\bibfnamefont {Michael~L.}\ \bibnamefont
  {Wall}}, \bibinfo {author} {\bibfnamefont {John~J.}\ \bibnamefont
  {Bollinger}}, \ and\ \bibinfo {author} {\bibfnamefont {Ana~Maria}\
  \bibnamefont {Rey}},\ }\bibfield  {title} {\enquote {\bibinfo {title}
  {Measuring out-of-time-order correlations and multiple quantum spectra in a
  trapped-ion quantum magnet},}\ }\href {http://dx.doi.org/10.1038/nphys4119}
  {\bibfield  {journal} {\bibinfo  {journal} {Nature Physics}\ }\textbf
  {\bibinfo {volume} {13}},\ \bibinfo {pages} {781–786} (\bibinfo {year}
  {2017})}\BibitemShut {NoStop}%
\bibitem [{\citenamefont {Li}\ \emph {et~al.}(2017)\citenamefont {Li},
  \citenamefont {Fan}, \citenamefont {Wang}, \citenamefont {Ye}, \citenamefont
  {Zeng}, \citenamefont {Zhai}, \citenamefont {Peng},\ and\ \citenamefont
  {Du}}]{Li_2017}%
  \BibitemOpen
  \bibfield  {author} {\bibinfo {author} {\bibfnamefont {Jun}\ \bibnamefont
  {Li}}, \bibinfo {author} {\bibfnamefont {Ruihua}\ \bibnamefont {Fan}},
  \bibinfo {author} {\bibfnamefont {Hengyan}\ \bibnamefont {Wang}}, \bibinfo
  {author} {\bibfnamefont {Bingtian}\ \bibnamefont {Ye}}, \bibinfo {author}
  {\bibfnamefont {Bei}\ \bibnamefont {Zeng}}, \bibinfo {author} {\bibfnamefont
  {Hui}\ \bibnamefont {Zhai}}, \bibinfo {author} {\bibfnamefont {Xinhua}\
  \bibnamefont {Peng}}, \ and\ \bibinfo {author} {\bibfnamefont {Jiangfeng}\
  \bibnamefont {Du}},\ }\bibfield  {title} {\enquote {\bibinfo {title}
  {Measuring out-of-time-order correlators on a nuclear magnetic resonance
  quantum simulator},}\ }\href {http://dx.doi.org/10.1103/PhysRevX.7.031011}
  {\bibfield  {journal} {\bibinfo  {journal} {Physical Review X}\ }\textbf
  {\bibinfo {volume} {7}},\ \bibinfo {pages} {031011} (\bibinfo {year}
  {2017})}\BibitemShut {NoStop}%
\bibitem [{\citenamefont {Wei}\ \emph {et~al.}(2019)\citenamefont {Wei},
  \citenamefont {Peng}, \citenamefont {Shtanko}, \citenamefont {Marvian},
  \citenamefont {Lloyd}, \citenamefont {Ramanathan},\ and\ \citenamefont
  {Cappellaro}}]{wei2019OTOC}%
  \BibitemOpen
  \bibfield  {author} {\bibinfo {author} {\bibfnamefont {Ken~Xuan}\
  \bibnamefont {Wei}}, \bibinfo {author} {\bibfnamefont {Pai}\ \bibnamefont
  {Peng}}, \bibinfo {author} {\bibfnamefont {Oles}\ \bibnamefont {Shtanko}},
  \bibinfo {author} {\bibfnamefont {Iman}\ \bibnamefont {Marvian}}, \bibinfo
  {author} {\bibfnamefont {Seth}\ \bibnamefont {Lloyd}}, \bibinfo {author}
  {\bibfnamefont {Chandrasekhar}\ \bibnamefont {Ramanathan}}, \ and\ \bibinfo
  {author} {\bibfnamefont {Paola}\ \bibnamefont {Cappellaro}},\ }\bibfield
  {title} {\enquote {\bibinfo {title} {Emergent prethermalization signatures in
  out-of-time ordered correlations},}\ }\href {\doibase
  10.1103/PhysRevLett.123.090605} {\bibfield  {journal} {\bibinfo  {journal}
  {Phys. Rev. Lett.}\ }\textbf {\bibinfo {volume} {123}},\ \bibinfo {pages}
  {090605} (\bibinfo {year} {2019})}\BibitemShut {NoStop}%
\bibitem [{\citenamefont {Li}\ \emph {et~al.}(2020)\citenamefont {Li},
  \citenamefont {Choudhury},\ and\ \citenamefont {Liu}}]{li2020fast}%
  \BibitemOpen
  \bibfield  {author} {\bibinfo {author} {\bibfnamefont {Zehan}\ \bibnamefont
  {Li}}, \bibinfo {author} {\bibfnamefont {Sayan}\ \bibnamefont {Choudhury}}, \
  and\ \bibinfo {author} {\bibfnamefont {W.~Vincent}\ \bibnamefont {Liu}},\
  }\href@noop {} {\enquote {\bibinfo {title} {Fast scrambling without appealing
  to holographic duality},}\ } (\bibinfo {year} {2020}),\ \Eprint
  {http://arxiv.org/abs/2004.11269} {arXiv:2004.11269 [cond-mat.quant-gas]}
  \BibitemShut {NoStop}%
\bibitem [{\citenamefont {Belyansky}\ \emph {et~al.}(2020)\citenamefont
  {Belyansky}, \citenamefont {Bienias}, \citenamefont {Kharkov}, \citenamefont
  {Gorshkov},\ and\ \citenamefont {Swingle}}]{belyansky2020minimal}%
  \BibitemOpen
  \bibfield  {author} {\bibinfo {author} {\bibfnamefont {Ron}\ \bibnamefont
  {Belyansky}}, \bibinfo {author} {\bibfnamefont {Przemyslaw}\ \bibnamefont
  {Bienias}}, \bibinfo {author} {\bibfnamefont {Yaroslav~A.}\ \bibnamefont
  {Kharkov}}, \bibinfo {author} {\bibfnamefont {Alexey~V.}\ \bibnamefont
  {Gorshkov}}, \ and\ \bibinfo {author} {\bibfnamefont {Brian}\ \bibnamefont
  {Swingle}},\ }\href@noop {} {\enquote {\bibinfo {title} {A minimal model for
  fast scrambling},}\ } (\bibinfo {year} {2020}),\ \Eprint
  {http://arxiv.org/abs/2005.05362} {arXiv:2005.05362 [quant-ph]} \BibitemShut
  {NoStop}%
\bibitem [{\citenamefont {Mehta}(2004)}]{mehta}%
  \BibitemOpen
  \bibfield  {author} {\bibinfo {author} {\bibfnamefont {M.~L.}\ \bibnamefont
  {Mehta}},\ }\href@noop {} {\emph {\bibinfo {title} {Random Matrices}}},\
  \bibinfo {edition} {3rd}\ ed.\ (\bibinfo  {publisher} {{Academic Press}},\
  \bibinfo {year} {2004})\BibitemShut {NoStop}%
\bibitem [{\citenamefont {Chen}\ and\ \citenamefont
  {Lucas}(2019{\natexlab{a}})}]{chen2019operator}%
  \BibitemOpen
  \bibfield  {author} {\bibinfo {author} {\bibfnamefont {Chi-Fang}\
  \bibnamefont {Chen}}\ and\ \bibinfo {author} {\bibfnamefont {Andrew}\
  \bibnamefont {Lucas}},\ }\href@noop {} {\enquote {\bibinfo {title} {Operator
  growth bounds from graph theory},}\ } (\bibinfo {year}
  {2019}{\natexlab{a}}),\ \Eprint {http://arxiv.org/abs/1905.03682}
  {arXiv:1905.03682 [math-ph]} \BibitemShut {NoStop}%
\bibitem [{\citenamefont {Chen}\ and\ \citenamefont
  {Lucas}(2019{\natexlab{b}})}]{chen2019finite}%
  \BibitemOpen
  \bibfield  {author} {\bibinfo {author} {\bibfnamefont {Chi-Fang}\
  \bibnamefont {Chen}}\ and\ \bibinfo {author} {\bibfnamefont {Andrew}\
  \bibnamefont {Lucas}},\ }\bibfield  {title} {\enquote {\bibinfo {title}
  {{Finite speed of quantum scrambling with long range interactions}},}\ }\href
  {\doibase 10.1103/PhysRevLett.123.250605} {\bibfield  {journal} {\bibinfo
  {journal} {Phys. Rev. Lett.}\ }\textbf {\bibinfo {volume} {123}},\ \bibinfo
  {pages} {250605} (\bibinfo {year} {2019}{\natexlab{b}})},\ \Eprint
  {http://arxiv.org/abs/1907.07637} {arXiv:1907.07637 [quant-ph]} \BibitemShut
  {NoStop}%
\bibitem [{\citenamefont
  {Lucas}(2019{\natexlab{a}})}]{lucas2019nonperturbative}%
  \BibitemOpen
  \bibfield  {author} {\bibinfo {author} {\bibfnamefont {Andrew}\ \bibnamefont
  {Lucas}},\ }\href@noop {} {\enquote {\bibinfo {title} {Non-perturbative
  dynamics of the operator size distribution in the {Sachdev-Ye-Kitaev}
  model},}\ } (\bibinfo {year} {2019}{\natexlab{a}}),\ \Eprint
  {http://arxiv.org/abs/1910.09539} {arXiv:1910.09539 [hep-th]} \BibitemShut
  {NoStop}%
\bibitem [{\citenamefont {Tran}\ \emph {et~al.}(2020)\citenamefont {Tran},
  \citenamefont {Chen}, \citenamefont {Ehrenberg}, \citenamefont {Guo},
  \citenamefont {Deshpande}, \citenamefont {Hong}, \citenamefont {Gong},
  \citenamefont {Gorshkov},\ and\ \citenamefont {Lucas}}]{tran2020hierarchy}%
  \BibitemOpen
  \bibfield  {author} {\bibinfo {author} {\bibfnamefont {Minh~C.}\ \bibnamefont
  {Tran}}, \bibinfo {author} {\bibfnamefont {Chi-Fang}\ \bibnamefont {Chen}},
  \bibinfo {author} {\bibfnamefont {Adam}\ \bibnamefont {Ehrenberg}}, \bibinfo
  {author} {\bibfnamefont {Andrew~Y.}\ \bibnamefont {Guo}}, \bibinfo {author}
  {\bibfnamefont {Abhinav}\ \bibnamefont {Deshpande}}, \bibinfo {author}
  {\bibfnamefont {Yifan}\ \bibnamefont {Hong}}, \bibinfo {author}
  {\bibfnamefont {Zhe-Xuan}\ \bibnamefont {Gong}}, \bibinfo {author}
  {\bibfnamefont {Alexey~V.}\ \bibnamefont {Gorshkov}}, \ and\ \bibinfo
  {author} {\bibfnamefont {Andrew}\ \bibnamefont {Lucas}},\ }\href@noop {}
  {\enquote {\bibinfo {title} {Hierarchy of linear light cones with long-range
  interactions},}\ } (\bibinfo {year} {2020}),\ \Eprint
  {http://arxiv.org/abs/2001.11509} {arXiv:2001.11509 [quant-ph]} \BibitemShut
  {NoStop}%
\bibitem [{\citenamefont {Nahum}\ \emph {et~al.}(2018)\citenamefont {Nahum},
  \citenamefont {Vijay},\ and\ \citenamefont {Haah}}]{nahum_operator_2018}%
  \BibitemOpen
  \bibfield  {author} {\bibinfo {author} {\bibfnamefont {Adam}\ \bibnamefont
  {Nahum}}, \bibinfo {author} {\bibfnamefont {Sagar}\ \bibnamefont {Vijay}}, \
  and\ \bibinfo {author} {\bibfnamefont {Jeongwan}\ \bibnamefont {Haah}},\
  }\bibfield  {title} {\enquote {\bibinfo {title} {Operator {Spreading} in
  {Random} {Unitary} {Circuits}},}\ }\href
  {https://link.aps.org/doi/10.1103/PhysRevX.8.021014} {\bibfield  {journal}
  {\bibinfo  {journal} {Physical Review X}\ }\textbf {\bibinfo {volume} {8}},\
  \bibinfo {pages} {021014} (\bibinfo {year} {2018})}\BibitemShut {NoStop}%
\bibitem [{\citenamefont {von Keyserlingk}\ \emph {et~al.}(2018)\citenamefont
  {von Keyserlingk}, \citenamefont {Rakovszky}, \citenamefont {Pollmann},\ and\
  \citenamefont {Sondhi}}]{von_keyserlingk_operator_2018}%
  \BibitemOpen
  \bibfield  {author} {\bibinfo {author} {\bibfnamefont {C.~W.}\ \bibnamefont
  {von Keyserlingk}}, \bibinfo {author} {\bibfnamefont {Tibor}\ \bibnamefont
  {Rakovszky}}, \bibinfo {author} {\bibfnamefont {Frank}\ \bibnamefont
  {Pollmann}}, \ and\ \bibinfo {author} {\bibfnamefont {S.~L.}\ \bibnamefont
  {Sondhi}},\ }\bibfield  {title} {\enquote {\bibinfo {title} {Operator
  {Hydrodynamics}, {OTOCs}, and {Entanglement} {Growth} in {Systems} without
  {Conservation} {Laws}},}\ }\href
  {https://link.aps.org/doi/10.1103/PhysRevX.8.021013} {\bibfield  {journal}
  {\bibinfo  {journal} {Physical Review X}\ }\textbf {\bibinfo {volume} {8}},\
  \bibinfo {pages} {021013} (\bibinfo {year} {2018})}\BibitemShut {NoStop}%
\bibitem [{\citenamefont {Roberts}\ \emph {et~al.}(2018)\citenamefont
  {Roberts}, \citenamefont {Stanford},\ and\ \citenamefont
  {Streicher}}]{Roberts:2018mnp}%
  \BibitemOpen
  \bibfield  {author} {\bibinfo {author} {\bibfnamefont {Daniel~A.}\
  \bibnamefont {Roberts}}, \bibinfo {author} {\bibfnamefont {Douglas}\
  \bibnamefont {Stanford}}, \ and\ \bibinfo {author} {\bibfnamefont
  {Alexandre}\ \bibnamefont {Streicher}},\ }\bibfield  {title} {\enquote
  {\bibinfo {title} {{Operator growth in the {SYK} model}},}\ }\href {\doibase
  10.1007/JHEP06(2018)122} {\bibfield  {journal} {\bibinfo  {journal} {JHEP}\
  }\textbf {\bibinfo {volume} {06}},\ \bibinfo {pages} {122} (\bibinfo {year}
  {2018})}\BibitemShut {NoStop}%
\bibitem [{\citenamefont {Lucas}(2019{\natexlab{b}})}]{lucas2019star}%
  \BibitemOpen
  \bibfield  {author} {\bibinfo {author} {\bibfnamefont {Andrew}\ \bibnamefont
  {Lucas}},\ }\href@noop {} {\enquote {\bibinfo {title} {Quantum many-body
  dynamics on the star graph},}\ } (\bibinfo {year} {2019}{\natexlab{b}}),\
  \Eprint {http://arxiv.org/abs/1903.01468} {arXiv:1903.01468
  [cond-mat.str-el]} \BibitemShut {NoStop}%
\bibitem [{\citenamefont {Maldacena}\ \emph {et~al.}(2016)\citenamefont
  {Maldacena}, \citenamefont {Shenker},\ and\ \citenamefont
  {Stanford}}]{Maldacena:2015waa}%
  \BibitemOpen
  \bibfield  {author} {\bibinfo {author} {\bibfnamefont {Juan}\ \bibnamefont
  {Maldacena}}, \bibinfo {author} {\bibfnamefont {Stephen~H.}\ \bibnamefont
  {Shenker}}, \ and\ \bibinfo {author} {\bibfnamefont {Douglas}\ \bibnamefont
  {Stanford}},\ }\bibfield  {title} {\enquote {\bibinfo {title} {{A bound on
  chaos}},}\ }\href {\doibase 10.1007/JHEP08(2016)106} {\bibfield  {journal}
  {\bibinfo  {journal} {JHEP}\ }\textbf {\bibinfo {volume} {08}},\ \bibinfo
  {pages} {106} (\bibinfo {year} {2016})}\BibitemShut {NoStop}%
\bibitem [{\citenamefont {Wang}\ and\ \citenamefont
  {Hazzard}(2019)}]{wang2019tightening}%
  \BibitemOpen
  \bibfield  {author} {\bibinfo {author} {\bibfnamefont {Zhiyuan}\ \bibnamefont
  {Wang}}\ and\ \bibinfo {author} {\bibfnamefont {Kaden R.~A.}\ \bibnamefont
  {Hazzard}},\ }\href@noop {} {\enquote {\bibinfo {title} {Tightening the
  {Lieb-Robinson} bound in locally-interacting systems},}\ } (\bibinfo {year}
  {2019}),\ \Eprint {http://arxiv.org/abs/1908.03997} {arXiv:1908.03997
  [quant-ph]} \BibitemShut {NoStop}%
\end{thebibliography}%

\end{document}